\documentclass[fleqn]{llncs}
\usepackage{pcfct}

\title{A Process Calculus with\\ Finitary Comprehended Terms}
\author{J.A. Bergstra \and C.A. Middelburg}
\institute{Informatics Institute, Faculty of Science,
           University of Amsterdam, \\
           Science Park~904, 1098~XH Amsterdam, the Netherlands \\
           \email{J.A.Bergstra@uva.nl, C.A.Middelburg@uva.nl}}

\begin{document}

\maketitle

\begin{abstract}
We introduce the notion of an \ACP\ process algebra and the notion of a
meadow enriched \ACP\ process algebra.
The former notion originates from the models of the axiom system \ACP.
The latter notion is a simple generalization of the former notion to
processes in which data are involved, the mathematical structure of data
being a meadow.
Moreover, for all associative operators from the signature of meadow
enriched \ACP\ process algebras that are not of an auxiliary nature, we 
introduce variable-binding operators as generalizations.
These variable-binding operators, which give rise to comprehended terms,
have the property that they can always be eliminated.
Thus, we obtain a process calculus whose terms can be interpreted in all
meadow enriched \ACP\ process algebras.
Use of the variable-binding operators can have a major impact on the
size of terms.
\begin{classcode}
D.1.3, F.1.2, F.4.1.
\end{classcode}
\end{abstract}

\section{Introduction}
\label{sect-introduction}

In many formalisms proposed for the description and analysis of
processes in which data are involved, algebraic specifications of
the data types concerned have to be given over and over again.
This is also the case with the principal \ACP-based formalisms proposed
for the description and analysis of processes in which data are
involved, to wit $\mu$CRL~\cite{GP94a,GP95a} and PSF~\cite{MV90a}.
There is a mismatch between the process specification part and the data
specification part of these formalisms.
Firstly, there is a choice of one built-in type of processes, whereas
there is a choice of all types of data that can be specified
algebraically.
Secondly, the semantics of the data specification part is its initial
algebra in the case of PSF and its class of minimal Boolean preserving
algebras in the case of $\mu$CRL, whereas the semantics of the process
specification part is a model based on transition systems and
bisimulation equivalence.
Sticking to this mismatch, no lasting axiomatizations in the style of
\ACP\ has emerged for process algebras that have to do with processes
in which data are involved.

Our first main objective is to obtain a lasting axiomatization in the
style of \ACP\ for process algebras that have to do with processes in
which data are involved.
To achieve this objective, we first introduce the notion of an \ACP\
process algebra and then the notion of a meadow enriched \ACP\ process
algebra.

\ACP\ process algebras are essentially models of the axiom system \ACP.
Meadow enriched \ACP\ process algebras are data enriched \ACP\ process
algebras in which the mathematical structure for data is a meadow.
Meadows were defined for the first time in~\cite{BT07a}.
The prime example of a meadow is the rational number field with the
multiplicative inverse operation made total by imposing that the
multiplicative inverse of zero is zero.
Although the notion of a meadow enriched \ACP\ process algebra is a
simple generalization of the notion of an \ACP\ process algebra, it is
an interesting one: there is a multitude of finite and infinite meadows
and meadows obviate the need for Boolean values and operations on data
that yield Boolean values to deal with conditions on data.

In the work on \ACP, the emphasis has always been on axiom systems.
In this paper, we put the emphasis on algebras.
That is, \ACP\ process algebras are looked upon in the same way as
groups, rings, fields, etc.\ are looked upon in universal algebra (see
e.g.~\cite{BS81a}).
The set of equations that are taken to characterize \ACP\ process
algebras is a revision of the axiom system \ACP.
The revision is primarily a matter of streamlining.
However, it also involves a minor generalization that allows for the
generalization to meadow enriched \ACP\ process algebras to proceed
smoothly.

In $\mu$CRL and PSF, we find variable-binding operators generalizing
associative operators of \ACP.
Our second main objective is to determine to what extent such
variable-binding operators fit in with meadow enriched \ACP\ process
algebras.
To achieve this objective, we introduce, for all associative operators
from the signature of meadow enriched \ACP\ process algebras that are 
not of an auxiliary nature, variable-binding operators as 
generalizations.

These variable-binding operators, which give rise to comprehended terms,
have the property that they can always be eliminated.
That is, for each comprehended term, we can derive from axioms
concerning the variable-binding operators that the comprehended term is
equal to a term over the signature of meadow enriched \ACP\ process
algebras.
Those axioms are axioms of a calculus because the distinction between
free and bound variables is essential in derivations.
The terms of this process calculus are interpreted in meadow enriched
\ACP\ process algebras.

Full elimination of all variable-binding operators occurring in a
comprehended term can lead to a combinatorial explosion.
We show that a combinatorial explosion can be prevented if
variable-binding operators that bind variables with a two-valued range
are still permitted in the resulting term.
We also show that in the latter case the size of the resulting term can
be further reduced if we add an identity element for sequential
composition to meadow enriched \ACP\ process algebras.
Moreover, we demonstrate that there is an alternative to introducing 
variable-binding operators for several associative operators on 
processes if we add a sort of process sequences and suitable operators 
on process sequences to meadow enriched \ACP\ process algebras.

For readability, it is imprecisely said above that the mathematical
structure for data in meadow enriched \ACP\ process algebras is a
meadow.
It is actually a signed meadow, i.e.\ a meadow expanded with a signum
operation.
In the presence of a signum operation, the ordering on the elements of a
meadow that corresponds with the usual ordering on the elements of a
field becomes definable.

This paper is organized as follows.
First, we give a brief summary of signed meadows
(Section~\ref{sect-meadows}).
Next, we introduce the notion of an \ACP\ process algebra
(Section~\ref{sect-ACP-process-alg}) and the notion of a meadow enriched
\ACP\ process algebra (Section~\ref{sect-MD-ACP-process-alg}).
After that, we associate a calculus with meadow enriched \ACP\ process
algebras (Section~\ref{sect-calculus-MD-ACP-process-alg}) and define the
interpretation of terms of this calculus in meadow enriched \ACP\
process algebras (Section~\ref{sect-interpretation-calculus}).
Following this, we investigate the consequences of elimination of
variable-binding operators from comprehended terms on the size of the
resulting terms (Section~\ref{sect-binary-var-binding-ops}).
Then, we investigate the effects of adding an identity element for
sequential composition to \ACP\ process algebras
(Section~\ref{sect-ACP-process-alg-id-seqc}) and the effects of adding
process sequences to \ACP\ process algebras
(Section~\ref{sect-ACP-process-alg-proc-vec}).
Finally, we make some concluding remarks
(Section~\ref{sect-conclusions}).

This paper consolidates material from~\cite{BM09b,BM09dd}.

\section{Signed Meadows}
\label{sect-meadows}

In this paper, the mathematical structure for data is a signed meadow.
In this section, we give a brief summary of signed meadows.

A meadow is a field with the multiplicative inverse operation made total
by imposing that the multiplicative inverse of zero is zero.
A signed meadow is a meadow expanded with a signum operation.
Meadows were defined for the first time in~\cite{BT07a} and were
investigated in e.g.~\cite{BBP13a,BHT09a,BM09g}.
The expansion of meadows with a signum operation originates
from~\cite{BBP13a}.

The signature of meadows is the same as the signature of fields.
It is a one-sorted signature.
We make the single sort explicit because we will extend this signature
to a two-sorted signature in Section~\ref{sect-MD-ACP-process-alg}.
The signature of meadows consists of the sort $\Quant$ of
\emph{quantities} and the following constants and operators:
\begin{itemize}
\item
the constants $\const{0}{\Quant}$ and $\const{1}{\Quant}$;
\item
the binary \emph{addition} operator
$\funct{+}{\Quant \x \Quant}{\Quant}$;
\item
the binary \emph{multiplication} operator
$\funct{\mul}{\Quant \x \Quant}{\Quant}$;
\item
the unary \emph{additive inverse} operator
$\funct{-}{\Quant}{\Quant}$;
\item
the unary \emph{multiplicative inverse} operator
$\funct{\minv}{\Quant}{\Quant}$.
\end{itemize}

We assume that there is a countably infinite set $\cU$ of variables of
sort $\Quant$, which contains $u$, $v$ and $w$, with and without
subscripts.
Terms are built as usual.
We use infix notation for the binary operators ${} + {}$ and
${} \mul {}$, prefix notation for the unary operator ${} -$, and postfix
notation for the unary operator ${}\minv$.
We use the usual precedence convention to reduce the need for
parentheses.
We introduce subtraction and division as abbreviations:
$p - q$ abbreviates $p + (-q)$ and
$p / q$ abbreviates $p \mul q\minv$.
For each non-negative natural number $n$, we write $\ul{n}$ for the
numeral for $n$.
That is, the term $\ul{n}$ is defined by induction on $n$ as follows:
$\ul{0} = 0$ and $\ul{n+1} = \ul{n} + 1$.
We also use the notation $p^n$ for exponentiation with a natural number
as exponent.
For each term $p$ over the signature of meadows, the term $p^n$ is
defined by induction on $n$ as follows: $p^0 = 1$ and
$p^{n+1} = p^n \mul p$.

The constants and operators from the signature of meadows are adopted
from rational arithmetic, which gives an appropriate intuition about
these constants and operators.

A \emph{meadow} is an algebra with the signature of meadows that 
satisfies the equations given in Table~\ref{eqns-meadow}.%
\begin{table}[!t]
\caption{Axioms for meadows}
\label{eqns-meadow}
\begin{eqntbl}
\begin{eqncol}
(u + v) + w = u + (v + w)                                             \\
u + v = v + u                                                         \\
u + 0 = u                                                             \\
u + (-u) = 0
\end{eqncol}
\qquad\quad
\begin{eqncol}
(u \mul v) \mul w = u \mul (v \mul w)                                 \\
u \mul v = v \mul u                                                   \\
u \mul 1 = u                                                          \\
u \mul (v + w) = u \mul v + u \mul w
\end{eqncol}
\qquad\quad
\begin{eqncol}
(u\minv)\minv = u                                                   \\
u \mul (u \mul u\minv) = u
\end{eqncol}
\end{eqntbl}
\end{table}
Thus, a meadow is a commutative ring with identity equipped with a
multiplicative inverse operation ${}\minv$ satisfying the reflexivity
equation $(u\minv)\minv = u $ and the restricted inverse equation
$u \mul (u \mul u\minv) = u$.
From the equations given in Table~\ref{eqns-meadow}, the equation
$0\minv = 0$ is derivable (see~\cite{BT07a}).

A \emph{non-trivial meadow} is a meadow that satisfies
the \emph{separation axiom}
\begin{ldispl}
0 \neq 1\;;
\end{ldispl}%
and a \emph{cancellation meadow} is a meadow that satisfies the
\emph{cancellation axiom}
\begin{ldispl}
u \neq 0 \And u \mul v = u \mul w \Implies v = w\;,
\end{ldispl}%
or equivalently, the \emph{general inverse law}
\begin{ldispl}
u \neq 0 \Implies u \mul u\minv = 1\;.
\end{ldispl}%

Important properties of non-trivial cancellation meadows are 
$u / u = 0 \Iff u = 0$ and $u / u = 1 \Iff u \neq 0$.
Henceforth, we will write $\cond{p}{r}{q}$ for
$(1 - r / r) \mul p + (r / r) \mul q$.
For non-trivial cancellation meadows, $\cond{p}{r}{q}$ can be read as 
follows: if $r$ equals $0$ then $p$ else~$q$.

Each field with the multiplicative inverse operation made total by 
imposing that the multiplicative inverse of zero is zero is a
non-trivial meadow.
The prime example of a non-trivial cancellation meadow is the rational 
number field with the multiplicative inverse operation made total by 
imposing that the multiplicative inverse of zero is zero.

A \emph{signed meadow} is a meadow expanded with a unary \emph{signum}
operation $\sign$ satisfying the equations given in
Table~\ref{eqns-signum}.%
\begin{table}[!t]
\caption{Additional axioms for signum operation}
\label{eqns-signum}
\begin{eqntbl}
\begin{eqncol}
\sign(u / u) = u / u                                                  \\
\sign(1 - u / u) = 1 - u / u                                          \\
\sign(-1) = -1
\end{eqncol}
\qquad\quad
\begin{eqncol}
\sign(u\minv) = \sign(u)                                              \\
\sign(u \mul v) = \sign(u) \mul \sign(v)                              \\
(1 - \frac{\sign(u) - \sign(v)}{\sign(u) - \sign(v)}) \mul
(\sign(u + v) - \sign(u)) = 0
\end{eqncol}
\end{eqntbl}
\end{table}
In combination with the cancellation axiom, the last equation in this
table is equivalent to the conditional equation
$\sign(u) = \sign(v) \Implies \sign(u + v) = \sign(u)$.
In signed meadows, the predicates $<$ and $>$ are defined as follows:
\begin{ldispl}
u < v \Iff 1 + \sign(u - v) = 0\;,
\\
u > v \Iff 1 - \sign(u - v) = 0\;.
\end{ldispl}%
In~\cite{BBP13a}, it is shown that the equational theories of signed
meadows and signed cancellation meadows are identical.

\section{ACP Process Algebras}
\label{sect-ACP-process-alg}

In this section, we introduce the notion of an \ACP\ process algebra.
This notion originates from the models of \ACP, an axiom system that was 
first presented in~\cite{BK84b}.
A comprehensive introduction to \ACP\ can be found in~\cite{BW90,Fok00}.

It is assumed that a fixed but arbitrary set $\Act$ of
\emph{atomic action names}, with $\dead \notin \Act$, has been given.

The signature of \ACP\ process algebras is a one-sorted signature.
We make the single sort explicit because we will extend this signature
to a two-sorted signature in Section~\ref{sect-MD-ACP-process-alg}.
The signature of \ACP\ process algebras consists of the sort $\Proc$ of
\emph{processes} and the following constants, operators, and predicate
symbols:
\begin{itemize}
\item
the \emph{deadlock} constant $\const{\dead}{\Proc}$;
\item
for each $e \in \Act$,
the \emph{atomic action} constant $\const{e}{\Proc}$;
\item
the binary \emph{alternative composition} operator
$\funct{\altc}{\Proc \x \Proc}{\Proc}$;
\item
the binary \emph{sequential composition} operator
$\funct{\seqc}{\Proc \x \Proc}{\Proc}$;
\item
the binary \emph{parallel composition} operator
$\funct{\parc}{\Proc \x \Proc}{\Proc}$;
\item
the binary \emph{left merge} operator
$\funct{\leftm}{\Proc \x \Proc}{\Proc}$;
\item
the binary \emph{communication merge} operator
$\funct{\commm}{\Proc \x \Proc}{\Proc}$;
\item
for each $H \subseteq \Act$,
the unary \emph{encapsulation} operator
$\funct{\encap{H}}{\Proc}{\Proc}$;
\item
the unary \emph{atomic action} predicate symbol
$\predt{\isact}{\Proc}$.
\end{itemize}

We assume that there is a countably infinite set $\cX$ of variables of
sort $\Proc$, which contains $x$, $y$ and $z$, with and without
subscripts.
Terms are built as usual.
We use infix notation for the binary operators.
We use the following precedence conventions to reduce the need for
parentheses: the operator $\altc$ binds weaker than all other binary
operators and the operator $\seqc$ binds stronger than all other binary
operators.

Let $P$ and $Q$ be closed terms of sort $\Proc$.
Intuitively, the constants, operators and predicate symbols introduced
above can be explained as follows:
\begin{itemize}
\item
$\dead$ is not capable of doing anything;
\item
$e$ is only capable of performing atomic action $e$ and next terminating
successfully;
\item
$P \altc Q$ behaves either as $P$ or as $Q$, but not both;
\item
$P \seqc Q$ first behaves as $P$ and on successful termination of $P$ it
next behaves as $Q$;
\item
$P \parc Q$ behaves as the process that proceeds with $P$ and $Q$ in
parallel;
\item
$P \leftm Q$ behaves the same as $P \parc Q$, except that it starts
with performing an atomic action of $P$;
\item
$P \commm Q$ behaves the same as $P \parc Q$, except that it starts with
performing an atomic action of $P$ and an atomic action of $Q$
synchronously;
\item
$\encap{H}(P)$ behaves the same as $P$, except that atomic actions from
$H$ are blocked;
\item
$\isact(P)$ holds if $P$ is an atomic action.
\end{itemize}
The operators $\leftm$ and $\commm$ are of an auxiliary nature. 
They are needed for the axiomatization of \ACP\ process algebras.

The predicate symbol $\isact$ is used to distinguish atomic actions 
from other processes. 
This predicate symbol, which does not occur in the axiom system \ACP, 
obviates the need to have a constant for each atomic action.
An alternative way to distinguish atomic actions from other processes is 
to have a subsort $\mathbf{A}$ of the sort $\Proc$.
We have not chosen this alternative way because it complicates matters
considerably.
Moreover, we prefer to keep close to elementary algebraic specification
(see e.g.~\cite{BT06a}).
By the notational convention introduced below, we seldom have to use
the predicate symbol $\isact$ explicitly.

In equations between terms of sort $\Proc$, we will use a notational
convention which requires the following assumption: there is a countably
infinite set $\cX' \subseteq \cX$ that contains $a$, $b$ and $c$, with
and without subscripts, but does not contain $x$, $y$ and $z$, with and
without subscripts.
Let $\phi$ be an equation between terms of sort $\Proc$, and let
$\set{a_1,\ldots,a_n}$ be the set of all variables from $\cX'$ that
occur in $\phi$.
Then we write $\phi$ for
$\isact(x_1) \And \ldots \And \isact(x_n) \Implies \phi'$, where
$\phi'$ is $\phi$ with, for all $i \in [1,n]$, all occurrences of $a_i$
replaced by $x_i$, and
$x_1,\ldots,x_n$ are variables from $\cX$ that do not occur in $\phi$.

An \emph{ACP process algebra} is an algebra with the signature of \ACP\
process algebras that satisfies the formulas given in
Table~\ref{eqns-ACP}.%
\begin{table}[!t]
\caption{Axioms for \ACP\ process algebras}
\label{eqns-ACP}
\begin{eqntbl}
\begin{eqncol}
x \altc y = y \altc x                                           \\
(x \altc y) \altc z = x \altc (y \altc z)                       \\
x \altc x = x                                                   \\
(x \altc y) \seqc z = x \seqc z \altc y \seqc z                 \\
(x \seqc y) \seqc z = x \seqc (y \seqc z)                       \\
x \altc \dead = x                                               \\
\dead \seqc x = \dead                                           \\
{}                                                              \\
{}                                                              \\
\encap{H}(e) = e     \hfill \mif e \notin H                     \\
\encap{H}(e) = \dead \hfill \mif e \in H                        \\
\encap{H}(x \altc y) = \encap{H}(x) \altc \encap{H}(y)          \\
\encap{H}(x \seqc y) = \encap{H}(x) \seqc \encap{H}(y)
\end{eqncol}
\qquad\quad
\begin{eqncol}
x \parc y =
        (x \leftm y \altc y \leftm x) \altc x \commm y          \\
a \leftm x = a \seqc x                                          \\
a \seqc x \leftm y = a \seqc (x \parc y)                        \\
(x \altc y) \leftm z = x \leftm z \altc y \leftm z              \\
a \commm b \seqc x = (a \commm b) \seqc x                       \\
a \seqc x \commm b \seqc y = (a \commm b) \seqc (x \parc y)     \\
(x \altc y) \commm z = x \commm z \altc y \commm z              \\
x \commm y = y \commm x                                         \\
(x \commm y) \commm z = x \commm (y \commm z)                   \\
\dead \commm x = \dead                                          \\
{}                                                              \\
\isact(e)                                                       \\
\isact(x) \And \isact(y) \Implies \isact(x \commm y)
\end{eqncol}
\end{eqntbl}
\end{table}
Three formulas in this table are actually schemas of formulas:
$e$ is a syntactic variable which stands for an arbitrary constant of
sort $\Proc$ (i.e.\ an atomic action constant or the deadlock constant).
A side condition is added to two schemas to restrict the constants for
which the syntactic variable stands.

Because the notational convention introduced above is used, the four
equations in Table~\ref{eqns-ACP} that are actually conditional
equations look the same as their counterpart in the axiom system \ACP.
It happens that these conditional equations allow for the generalization
to meadow enriched \ACP\ process algebras to proceed smoothly.
Apart from this, the set of formulas given in Table~\ref{eqns-ACP}
differs from the axiom system \ACP\ on three points.
Firstly, the equations $x \commm y = y \commm x$,
$(x \commm y) \commm z = x \commm (y \commm z)$, and
$\dead \commm x = \dead$ have been added.
In the axiom system \ACP, all closed substitution instances of these
equations are derivable.
Secondly, the equations $a \seqc x \commm b = (a \commm b) \seqc x$
and $x \commm (y \altc z) = x \commm\nolinebreak y \altc x \commm z$ 
have been removed.
These equations can be derived using the added equation
$x \commm y = y \commm x$.
Thirdly, the formulas $\isact(e)$ and
$\isact(x) \And \isact(y) \Implies \isact(x \commm y)$ have been added.
They express that the processes denoted by constants of sort $\Proc$ are
atomic actions and that the processes that result from the communication
merge of two atomic actions are atomic actions.
This does not exclude that there are additional atomic actions, which is
impossible in the case of \ACP.

For each model of the axiom system \ACP\ given in~\cite{BW90}, its 
expansion with the appropriate interpretation of the atomic action 
predicate symbol $\isact$ is an \ACP\ process algebra.

Not all processes in an \ACP\ process algebra have to be interpretations
of closed terms, even if all atomic actions are interpretations of
closed terms.
The processes concerned may be solutions of sets of recursion equations.
It is recommendable to restrict the attention to \ACP\ process algebras
satisfying additional axioms by which sets of recursion equations that
fulfil a guardedness condition have unique solutions.
For a comprehensive treatment of this issue, the reader is referred
to~\cite{BW90}.

\section{Meadow Enriched ACP Process Algebras}
\label{sect-MD-ACP-process-alg}

In this section, we introduce the notion of a meadow enriched \ACP\
process algebra.
This notion is a simple generalization of the notion of an \ACP\
process algebra introduced in Section~\ref{sect-ACP-process-alg} to
processes in which data are involved.
The elements of a signed meadow are taken as data.

The signature of meadow enriched \ACP\ process algebras is a two-sorted
signature.
It consists of the sorts, constants and operators from the signatures of
\ACP\ process algebras and signed meadows and in addition the following
operators:
\begin{itemize}
\item
for each $n \in \Nat$ and $e \in \Act$,
the $n$-ary \emph{data handling atomic action} operator
$\funct{e}
 {\underbrace{\Quant \x \cdots \x \Quant}_{n\; \mathrm{times}}}{\Proc}$;
\item
the binary \emph{guarded command} operator
$\funct{\gc}{\Quant \x \Proc}{\Proc}$.
\end{itemize}

We take the variables in $\cU$ for the variables of sort $\Quant$ and
the variables in $\cX$ for the variables of sort $\Proc$.
We assume that the sets $\cU$ and $\cX$ are disjoint.
Terms are built as usual for a many-sorted signature
(see e.g.~\cite{ST99a,Wir90a}).
We use the same notational conventions as before.
In addition, we use infix notation for the binary operator ${} \gc {}$.

Let $p_1,\ldots,p_n$ and $p$ be closed terms of sort $\Quant$ and $P$ be
a closed term of sort $\Proc$.
Intuitively, the additional operators can be explained as follows:
\begin{itemize}
\item
$e(p_1,\ldots,p_n)$ is only capable of performing data handling atomic
action $e(p_1,\ldots,p_n)$ and next terminating successfully;
\item
$p \gc P$ behaves as the process $P$ if $p$ equals $0$ and is not
capable of doing anything otherwise.
\end{itemize}

The different guarded command operators that have been proposed before
in the setting of \ACP\ have one thing in common: their first operand is
considered to stand for an element of the domain of a Boolean algebra
(see e.g.~\cite{BM05a}).
In contrast with those guarded command operators, the first operand of
the guarded command operator introduced here is considered to stand for
an element of the domain of a signed meadow.

A \emph{meadow enriched ACP process algebra} is an algebra with the 
signature of meadow enriched \ACP\ process algebras that satisfies the 
formulas given in Tables~\ref{eqns-meadow}--\ref{eqns-ACPmd}.%
\begin{table}[!t]
\caption{Additional axioms for meadow enriched \ACP\ process algebras}
\label{eqns-ACPmd}
\begin{eqntbl}
\begin{array}{@{}l@{}}
\begin{eqncol}
0 \gc x = x                                                       \\
1 \gc x = \dead                                                   \\
u \gc x = (u / u) \gc x                                           \\
u \gc (v \gc x) = (1 - (1 - u/u) \mul (1 - v/v)) \gc x            \\
u \gc x \altc v \gc x = (u/u \mul v/v) \gc x                      \\
\end{eqncol}
\qquad\quad
\begin{eqncol}
u \gc \dead = \dead                                               \\
u \gc (x \altc y) = u \gc x \altc u \gc y                         \\
u \gc x \seqc y = (u \gc x) \seqc y                               \\
(u \gc x) \leftm y = u \gc (x \leftm y)                           \\
(u \gc x) \commm y = u \gc (x \commm y)                           \\
\encap{H}(u \gc x) = u \gc \encap{H}(x)              \vspace*{1.5ex}              
\end{eqncol}
\\
\begin{ceqncol}
e \commm e' = e'' \Implies {} \\
\multicolumn{2}{@{}l@{}}{\;\;
e(u_1,\ldots,u_n) \commm e'(v_1,\ldots,v_n) =
(u_1 - v_1) \gc
(\cdots \gc ((u_n - v_n) \gc e''(u_1,\ldots,u_n))\cdots)}         \\
e \commm e' = \dead\phantom{'} \Implies
e(u_1,\ldots,u_n) \commm e'(v_1,\ldots,v_n) = \dead               \\
e(u_1,\ldots,u_n) \commm e'(v_1,\ldots,v_m) = \dead & \mif n \neq m
\eqnsep
\encap{H}(e(u_1,\ldots,u_n)) = e(u_1,\ldots,u_n)
 & \mif e \not\in H                                               \\
\encap{H}(e(u_1,\ldots,u_n)) = \dead
 & \mif e \in H
\eqnsep
\isact(e(u_1,\ldots,u_n))
\end{ceqncol}
\end{array}
\end{eqntbl}
\end{table}
Like in Table~\ref{eqns-ACP}, some formulas in
Table~\ref{eqns-ACPmd} are actually schemas of formulas: $e$, $e'$ and 
$e''$ are syntactic variables which stand for arbitrary constants of 
sort $\Proc$ different from $\dead$ and, in addition, $n$ and $m$ stand 
for arbitrary natural numbers.

For meadow enriched \ACP\ process algebras that satisfy the separation
axiom and the cancellation axiom, the five equations concerning the 
guarded command operator on the left-hand side in the upper half of 
Table~\ref{eqns-ACPmd} can easily be understood by taking the view that 
$0$ and $1$ represent the Boolean values $\True$ and $\False$, 
respectively.
In that case, we have that
\begin{itemize}
\item
$p/p$ models the test that yields $\True$ if $p = 0$ and $\False$
otherwise;
\item
if both $p$ and $q$ are equal to $0$ or $1$, then
$1 - p$ models $\Not p$, $p \mul q$ models $p \Or q$, and
consequently $1 - (1 - p) \mul (1 - q)$ models $p \And q$.
\end{itemize}
From this view, the equations given in the upper half of
Table~\ref{eqns-ACPmd} differ from the axioms for the most
general kind of guarded command operator that has been proposed in the
setting of \ACP\ (see e.g.~\cite{BM05a}) on two points only.
Firstly, the equation $u \gc x = u / u \gc x$ has been added.
This equation formalizes the informal explanation of the guarded command
given above.
Secondly, the equation $x \commm (u \gc y) = u \gc (x \commm y)$ has
been removed.
This equation can be derived using the equation
$x \commm y = y \commm x$ from Table~\ref{eqns-ACP}.

The equations in Table~\ref{eqns-ACPmd} concerning the communication
merge of data handling atomic actions formalize the intuition that two
data handling atomic actions $e(p_1,\ldots,p_n)$ and
$e'(q_1,\ldots,q_m)$ can be performed synchronously iff 
$e$ and $e'$ can be performed synchronously and $n = m$ and $p_1 = q_1$ 
and $\ldots$ and $p_n = q_n$.
The equations concerning the encapsulation of data handling atomic
actions agree with the way in which the encapsulation of data handling 
atomic actions is dealt with in $\mu$CRL and PSF.
The formula concerning the atomic action predicate simply expresses that
data handling atomic actions are also atomic actions.

Henceforth, we will write $\cond{P}{p}{Q}$ for
$(p / p) \gc P \altc (1 - p / p) \gc Q$.
For meadow enriched \ACP\ process algebras that satisfy the separation
axiom and the cancellation axiom, $\cond{P}{p}{Q}$ can be read as 
follows: if $p$ equals $0$ then $P$ else $Q$.

For each \ACP\ process algebra $\gA'$ and each signed non-trivial 
cancellation meadow $\gA''$, there exists an amalgamation of $\gA'$ and 
$\gA''$, i.e.\ a model of the axioms for both \ACP\ process algebras and 
signed non-trivial cancellation meadows whose restriction to the 
signature of \ACP\ process algebras is $\gA'$ and whose restriction to 
the signature of signed meadows is $\gA''$ (by the amalgamation result 
about expansions presented as Theorem~6.1.1 in~\cite{Hod93a}, adapted to 
the many-sorted case).      
For each amalgamation of an \ACP\ process algebra with a countably
infinite set of atomic actions and a signed non-trivial cancellation 
meadow, its expansion with the appropriate interpretation of the data 
handling atomic action operators $e$ and the guarded command operator 
$\gc$ is a meadow enriched \ACP\ process algebra.

In subsequent sections, we write $\Sigmp$ for the signature of meadow
enriched \ACP\ process algebras.

\section{A Calculus for Meadow Enriched ACP Process~Algebras}
\label{sect-calculus-MD-ACP-process-alg}

In this section, we associate a calculus with meadow enriched \ACP\
process algebras.
For that, we introduce, for all associative operators from the signature
of meadow enriched \ACP\ process algebras that are not of an auxiliary 
nature, variable-binding operators as generalizations.
To build terms of the calculus, called binding terms, both the constants
and operators from the signature of meadow enriched \ACP\ process
algebras and those variable-binding operators are available.

The sets of \emph{binding terms} of sorts $\Quant$ and $\Proc$, written
$\BT{\Quant}$ and $\BT{\Proc}$, respectively, are inductively defined by
the following formation rules (where $S_1,\ldots,S_n$ and $S$ range over 
the sorts from $\Sigmp$):
\begin{itemize}
\item
if $u \in \cU$, then $u \in \BT{\Quant}$;
\item
if $x \in \cX$, then $x \in \BT{\Proc}$;
\item
if $\const{c}{S}$ is a constant from $\Sigmp$, then $c \in \BT{S}$;
\item
if $\funct{o}{S_1 \x \cdots \x S_n}{S}$ is an operator from $\Sigmp$ and
$t_1 \in \BT{S_1}$, \ldots, $t_n \in \BT{S_n}$, then
$o(t_1,\ldots,t_n) \in \BT{S}$;
\item
if $u \in \cU$ and $t \in \BT{\Quant}$, then, for each $n \in \Natpos$,
$\Sum{n}{u} t \in \BT{\Quant}$ and $\Prod{n}{u} t \in \BT{\Quant}$;%
\footnote{We write $\Natpos$ for the set $\Nat \diff \set{0}$.}
\item
if $u \in \cU$ and $t \in \BT{\Proc}$, then, for each $n \in \Natpos$,
$\Chc{n}{u} t \in \BT{\Proc}$, $\Seq{n}{u} t \in \BT{\Proc}$, and
$\Par{n}{u} t \in \BT{\Proc}$.
\end{itemize}

$\Sum{n}{}$, $\Prod{n}{}$, $\Chc{n}{}$, $\Seq{n}{}$, and $\Par{n}{}$ are
the variable-binding operators mentioned above.
They bind variables that range over all quantities that can be denoted
by numerals $\ul{k}$ where $0 \leq k < n$ (in plain terms, quantities 
that correspond to natural numbers less than $n$).
Intuitively, $\Sum{n}{u} t$ stands for $t_1 + \cdots + t_n$, where $t_i$ 
($1 \leq i \leq n$) is $t$ with all occurrences of $u$ replaced by 
$\ul{u-1}$, and analogously in the case of $\Prod{n}{}$, $\Chc{n}{}$, 
$\Seq{n}{}$, and $\Par{n}{}$.

A binding term $t$ is a \emph{comprehended term} if it is a binding term
of the form $\Bnd{n}{u} t'$, where $\Bnd{n}{}$ is a variable-binding
operator.%
\footnote
{The name comprehended term originates from the name comprehended
 expression introduced in~\cite{RLG92a}.
}
Below, we will give the axioms of the calculus associated with meadow
enriched \ACP\ process algebras.
We have to do with a calculus because the distinction between free and
bound variables is essential in applying the axioms concerning
comprehended terms.

A variable $u \in \cU$ occurs \emph{free} in a binding term $t$
if there is an occurrence of $u$ in $t$ that is not in a subterm of the
form $\Bnd{n}{u} t'$, where $\Bnd{n}{}$ is a variable-binding operator.
A binding term $t$ is \emph{closed} if it is a binding term in which no
variable occurs free.

Substitution of a binding term $t'$ of sort $\Proc$ for a variable
$x \in \cX$ in a binding term $t$, written $t \subst{t'}{x}$, is defined
by induction on the structure of $t$ as usual:
\pagebreak[2]
\begin{ldispl}
\begin{caeqns}
v \subst{t'}{x} & = & v\;, \\
y \subst{t'}{x} & = &
\left\{\begin{col}
       t' \\ y
       \end{col}
\right.
       &
       \begin{col}
       \mif x \equiv y\;,\footnotemark \\ \mother\;,
       \end{col}
\\
c \subst{t'}{x} & = & c\;, \\
o(t_1,\ldots,t_n) \subst{t'}{x} & = &
o(t_1 \subst{t'}{x},\ldots,t_n \subst{t'}{x})\;, \\[.75ex]
(\Bnd{n}{v} t'') \subst{t'}{x} & = &
\left\{\begin{col}
       \Bnd{n}{w} ((t'' \subst{w}{v}) \subst{t'}{x}) \\[3ex]
       \Bnd{n}{v} (t'' \subst{t'}{x})
       \end{col}
\right.
       &
       \begin{col}
       \mif v \mathrm{\;occurs\;free\;in\;} t' \\
       (w \mathrm{\;does\;not\;occur\;in\;} t',t'')\;, \\
       \mother\;.
       \end{col}
\end{caeqns}
\end{ldispl}%
\footnotetext{We write $\equiv$ for syntactic identity.}%
and substitution of a binding term $t'$ of sort $\Quant$ for a variable
$u \in \cU$ in a binding term $t$, written $t \subst{t'}{u}$, is defined
by induction on the structure of $t$ as follows:
\begin{ldispl}
\begin{caeqns}
v \subst{t'}{u} & = &
\left\{\begin{col}
       t' \\ v
       \end{col}
\right.
       &
       \begin{col}
       \mif u \equiv v\;, \\ \mother\;,
       \end{col}
\\
x \subst{t'}{u} & = & x\;, \\
c \subst{t'}{u} & = & c\;, \\
o(t_1,\ldots,t_n) \subst{t'}{u} & = &
o(t_1 \subst{t'}{u},\ldots,t_n \subst{t'}{u})\;, \\[.75ex]
(\Bnd{n}{v} t'') \subst{t'}{u} & = &
\left\{\begin{col}
       \Bnd{n}{v} t'' \eqnsep
       \Bnd{n}{w} ((t'' \subst{w}{v}) \subst{t'}{u}) \eqnsep
       \Bnd{n}{v} (t'' \subst{t'}{u})
       \end{col}
\right.
       &
       \begin{col}
       \mif u \equiv v\;, \\
       \mif u \not\equiv v,
            v \mathrm{\;occurs\;free\;in\;} t' \\
       (w \mathrm{\;does\;not\;occur\;in\;} t',t'')\;, \\
       \mother\;.
       \end{col}
\end{caeqns}
\end{ldispl}%

The essentiality of the distinction between free and bound variables in
applying the axioms concerning comprehended terms originates from the
substitutions involved in applying those axioms.

The axioms of the calculus associated with meadow enriched \ACP\ process
algebras are the formulas given in
Tables~\ref{eqns-meadow}--\ref{eqns-compr-terms}.%
\begin{table}[!t]
\caption{Axioms for comprehended terms}
\label{eqns-compr-terms}
\begin{leqntbl}
\begin{eqncol}
\Sum{n}{u} p  = \Sum{n}{v} (p \subst{v}{u}) 
 \hfill \;\; \mif v \mathrm{\;does\;not\;occur\;free\;in\;} p         \\
\Sum{1}{u} p  = p \subst{0}{u}                                        \\
\Sum{n+1}{u} p  = p \subst{0}{u} + \Sum{n}{u} (p \subst{u + 1}{u})
\eqnsep
\Prod{n}{u} p = \Prod{n}{v} (p \subst{v}{u})                          
 \hfill \;\; \mif v \mathrm{\;does\;not\;occur\;free\;in\;} p         \\
\Prod{1}{u} p = p \subst{0}{u}                                        \\
\Prod{n+1}{u} p = p \subst{0}{u} \mul \Prod{n}{u} (p \subst{u + 1}{u})
\eqnsep
\Chc{n}{u} P  = \Chc{n}{v} (P \subst{v}{u})                           
 \hfill \;\; \mif v \mathrm{\;does\;not\;occur\;free\;in\;} P         \\
\Chc{1}{u} P  = P \subst{0}{u}                                        \\
\Chc{n+1}{u} P = P \subst{0}{u} \altc \Chc{n}{u} (P \subst{u + 1}{u})
\eqnsep
\Seq{n}{u} P  = \Seq{n}{v} (P \subst{v}{u})                           
 \hfill \;\; \mif v \mathrm{\;does\;not\;occur\;free\;in\;} P         \\
\Seq{1}{u} P  = P \subst{0}{u}                                        \\
\Seq{n+1}{u} P = P \subst{0}{u} \seqc \Seq{n}{u} (P \subst{u + 1}{u})
\eqnsep
\Par{n}{u} P  = \Par{n}{v} (P \subst{v}{u})                           
 \hfill \;\; \mif v \mathrm{\;does\;not\;occur\;free\;in\;} P         \\
\Par{1}{u} P  = P \subst{0}{u}                                        \\
\Par{n+1}{u} P = P \subst{0}{u} \parc \Par{n}{u} (P \subst{u + 1}{u})
\end{eqncol}
\end{leqntbl}
\end{table}
Like some equations in Tables~\ref{eqns-ACP} and~\ref{eqns-ACPmd}, the
equations in Table~\ref{eqns-compr-terms} are actually schemas of
equations: $p$ and $P$ are syntactic variables which stand for arbitrary
binding terms of sort $\Quant$ and sort $\Proc$, respectively, and $n$
stands for an arbitrary positive natural number.

The axioms given in Table~\ref{eqns-compr-terms} are called the
\emph{axioms for comprehended terms}.
They consist of three axioms, including an $\alpha$-conversion axiom,
for each of the variable-binding operators of the calculus.
For each comprehended term, we can derive from these axioms that the
comprehended term is equal to a term over the signature of meadow
enriched \ACP\ process algebras.
\begin{theorem}[Elimination]
\label{theorem-elim-binders}
For all comprehended terms $t$, there exists a term $t'$ over the
signature of meadow enriched \ACP\ process algebras such that $t = t'$
is derivable from the axioms for comprehended terms.
\end{theorem}
\begin{proof}
If $t$ is of the form $\Sum{n}{u} t''$, $\Prod{n}{u} t''$,
$\Chc{n}{u} t''$, $\Seq{n}{u} t''$ or $\Par{n}{u} t''$, where $t''$ is a
term over the signature of meadow enriched \ACP\ process algebras of the
right sort, then it is easy to prove by induction on $n$ that there
exists a term $t'$ over the signature of meadow enriched \ACP\ process
algebras such that $t = t'$ is derivable from the axioms for
comprehended terms.
Using this fact, the general case is easily proved by induction on the
depth of $t$.
\qed
\end{proof}

The comprehended terms of the calculus associated with meadow enriched
\ACP\ process algebras are \emph{finitary} comprehended terms because
the variable-binding operators of the calculus bind variables with a
finite range only.
This is a prerequisite for elimination of variable-binding operators.

\section{The Interpretation of Terms of the Calculus}
\label{sect-interpretation-calculus}

In this section, we define the interpretation of terms of the calculus
associated with meadow enriched \ACP\ process algebras.
We assume that a fixed but arbitrary meadow enriched \ACP\ process
algebra $\gA$ has been given.

We write $\sigma_\gA$, where $\sigma$ in $\Sigmp$, for the
interpretation of $\sigma$ in $\gA$.
Moreover,
we write $f + 1$, where $\funct{f}{\Quant_\gA}{\Quant_\gA}$ or
$\funct{f}{\Quant_\gA}{\Proc_\gA}$,
for the function $\funct{f'}{\Quant_\gA}{\Quant_\gA}$ or
$\funct{f'}{\Quant_\gA}{\Proc_\gA}$, respectively, defined by
$f'(q) = f(q +_\gA 1_\gA)$.

The terms of the calculus introduced above can be directly interpreted 
in $\gA$.
To achieve that, we associate with each variable-binding operator
$\Bnd{n}{}$ of the calculus a function
$\funct{\Bnd{n}{\gA}}{(\Quant_\gA \to \Quant_\gA)}{\Quant_\gA}$ or
$\funct{\Bnd{n}{\gA}}{(\Quant_\gA \to \Proc_\gA)}{\Proc_\gA}$
as follows:
\begin{ldispl}
\begin{aeqns}
\Sum{1}{\gA}(f)   & = & f(0_\gA)\;, \\
\Sum{n+1}{\gA}(f) & = & f(0_\gA) +_\gA \Sum{n}{\gA}(f + 1)\;,
\eqnsep
\Prod{1}{\gA}(f)   & = & f(0_\gA)\;, \\
\Prod{n+1}{\gA}(f) & = & f(0_\gA) \mul_\gA \Prod{n}{\gA}(f + 1)\;,
\end{aeqns}
\quad\;\;
\begin{aeqns}
\Chc{1}{\gA}(f)   & = & f(0_\gA)\;, \\
\Chc{n+1}{\gA}(f) & = & f(0_\gA) \altc_\gA \Chc{n}{\gA}(f + 1)\;,
\eqnsep
\Seq{1}{\gA}(f)   & = & f(0_\gA)\;, \\
\Seq{n+1}{\gA}(f) & = & f(0_\gA) \seqc_\gA \Seq{n}{\gA}(f + 1)\;,
\eqnsep
\Par{1}{\gA}(f)   & = & f(0_\gA)\;, \\
\Par{n+1}{\gA}(f) & = & f(0_\gA) \parc_\gA \Par{n}{\gA}(f + 1)\;.
\end{aeqns}
\end{ldispl}%

The interpretation of a term of the calculus in $\gA$ depends on the
elements of $\Quant_\gA$ and $\Proc_\gA$ that are associated with the
variables that occur free in it.
We model such associations by functions
$\funct{\rho}{(\cU \union \cX)}{(\Quant_\gA \union \Proc_\gA)}$ such
that $u \in \cU \Implies \rho(u) \in \Quant_\gA$ and
$x \in \cX \Implies \rho(x) \in \Proc_\gA$.
These functions are called \emph{assignments} in $\gA$.
We write $\Ass{\gA}$ for the set of all assignments in $\gA$.
For each assignment $\rho \in \Ass{\gA}$, $u \in \cU$ and
$q \in \Quant_\gA$, we write $\rho(u \to q)$ for the unique assignment
$\rho' \in \Ass{\gA}$ such that $\rho'(v) = \rho(v)$ if $v \not\equiv u$
and $\rho'(u) = q$.

The interpretation of terms of the calculus in a meadow enriched \ACP\
process algebra $\gA$ is given by the function
$\funct{\Int{\gA}{\ph}}{(\BT{\Quant} \union \BT{\Proc})}
                       {(\Ass{\gA} \to (\Quant_\gA \union \Proc_\gA))}$
defined as follows:
\begin{ldispl}
\begin{aeqns}
\Int{\gA}{u}(\rho) & = & \rho(u)\;, \\
\Int{\gA}{x}(\rho) & = & \rho(x)\;, \\
\Int{\gA}{c}(\rho) & = & c_\gA\;, \\
\Int{\gA}{o(t_1,\ldots,t_n)}(\rho) & = &
o_\gA(\Int{\gA}{t_1}(\rho),\ldots,\Int{\gA}{t_n}(\rho))\;, \\
\Int{\gA}{\Bnd{n}{u} t}(\rho) & = & \Bnd{n}{\gA}(f),
\mathrm{\,where\,} f \mathrm{\,is\,defined\,by\,}
f(q) = \Int{\gA}{t}(\rho(u \to q))\;.
\end{aeqns}
\end{ldispl}%

The axioms of the calculus associated with meadow enriched \ACP\ process
algebras are sound with respect to the interpretation of the terms of
the calculus given above.
\begin{theorem}[Soundness]
\label{theorem-soundness}
For all equations $t = t'$ that belong to the axioms of the calculus
associated with meadow enriched \ACP\ process algebras, we have that
$\Int{\gA}{t}(\rho) = \Int{\gA}{t'}(\rho)$ for all assignments
$\rho \in \Ass{\gA}$.
\end{theorem}
\begin{proof}
For all equations $t = t'$ that belong to the axioms for meadow enriched
\ACP\ process algebras, the soundness follows immediately from the fact
that $\gA$ is a meadow enriched \ACP\ process algebra.
For all equations $t = t'$ that belong to the axioms for comprehended
terms, the soundness is easily proved by induction on the structure of
\nolinebreak[2] $t$.
\qed
\end{proof}

Because the terms of the calculus associated with meadow enriched \ACP\ 
process algebras can be directly interpreted in meadow enriched \ACP\ 
process algebras, we consider the variable-binding operators of the 
calculus to constitute a process algebraic feature.
Fitting them in an algebraic framework does not involve any serious 
theoretical complication.
It is much more difficult to fit the variable-binding operators from
$\mu$CRL and PSF that generalize associative operators of \ACP, but 
do not give rise to finitary comprehended terms, in an algebraic 
framework (see e.g.~\cite{Lut02a}).

\section{The Binary Variable-Binding Operators}
\label{sect-binary-var-binding-ops}

Full elimination of all variable-binding operators occurring in a
comprehended term can lead to a combinatorial explosion.
In this section, we show that no combinatorial explosion takes place if
variable-binding operators that bind variables with a two-valued range
are still permitted in the resulting term.

We begin by looking at an example.
From the axioms for comprehended terms, we easily derive the equation
\begin{ldispl}
\Sum{7}{u} p = p \subst{\ul{0}}{u} + \cdots + p \subst{\ul{6}}{u}\;.%
\end{ldispl}%
This suggests that, on full elimination of variable-binding operators, 
the size of the resulting term grows rapidly as the size of the original
term increases (there are seven substitution instances of $p$ and they
have increasing sizes).
Using the axioms for comprehended terms as well as other axioms of the
calculus, we derive the following:
\begin{ldispl}
\begin{aeqns}
   &   &
p \subst{\ul{0}}{u} + \cdots + p \subst{\ul{6}}{u}
\\ & = &
p \subst{\ul{0}}{u} + \cdots + p \subst{\ul{6}}{u} + 0
\\ & = &
(\cond{0}{1 - \sign(u - \ul{6})}{p}) \subst{\ul{0}}{u} + \cdots +
(\cond{0}{1 - \sign(u - \ul{6})}{p}) \subst{\ul{7}}{u}
\\ & = &
\Sum{2}{u}
 \bigl(\Sum{2}{v}
  \bigl(\Sum{2}{w}
   \bigl((\cond{0}{1 - \sign(u - \ul{6})}{p})
          \subst{\ul{2}^2 \mul w + \ul{2}^1 \mul v + \ul{2}^0 \mul u}{u}
   \bigr)\bigr)\bigr)
\\ & = &
\Sum{2}{u}
 \bigl(\Sum{2}{v}
  \bigl(\Sum{2}{w}
   \bigl(((\cond{0}{1 - \sign(u - \ul{6})}{p})
          \subst{\ul{2} \mul v + u}{u}) \subst{\ul{2} \mul w + v}{v}
   \bigr)\bigr)\bigr)\;.
\end{aeqns}
\end{ldispl}%
This suggests that, if variable-binding operators that bind variables 
with a two-valued range are still permitted in the resulting term, its 
size grows far less rapidly as the size of the original term increases
(there is only one substitution instance of $p$).
However, a counterpart of the first step in the derivation above does
not exist for comprehended terms of the forms $\Seq{n}{u} p$ and
$\Par{n}{u} p$ because identity elements for sequential and parallel 
composition are missing.

Henceforth, we will use the term
\emph{binary variable-binding operators} for the variable-binding
operators that bind variables with a two-valued range and the term
\emph{non-binary variable-binding operators} for the other
variable-binding operators.

The size of binding terms is given by the function
$\funct{\tsize}{(\BT{\Quant} \union \BT{\Proc})}{\Nat}$
defined as follows:
\begin{ldispl}
\begin{aeqns}
\tsize(u) & = & 1\;, \\
\tsize(x) & = & 1\;, \\
\tsize(c) & = & 1\;, \\
\tsize(o(t_1,\ldots,t_n)) & = &
\tsize(t_1) + \cdots + \tsize(t_n) + 1\;, \\
\tsize\bigl(\Bnd{n}{u}(t)\bigr) & = & \tsize(t) + \logii(n) + 1\;.%
\footnotemark
\end{aeqns}
\end{ldispl}%
\footnotetext
{We use the convention that, whenever we write $\logii(n)$ in a context
 requiring a natural number, $\lceil \logii(n) \rceil$ is implicitly
 meant.
}%
The summand $\logii(n)$ occurs in the equation for the size of a term of
the form $\Bnd{n}{u}(t)$ because having (the cardinality of) the range
of $u$ encoded in the variable-binding operator is an artifice that must
be taken into account using the most efficient way in which $\ul{n}$
could be represented by a binding term.
It follows from Proposition~\ref{prop-eqns-derived} formulated below
that the size of this term is of order $\logii(n)$.

The important insights relevant to elimination of non-binary
variable-binding operators are brought together in the following
proposition.
\begin{proposition}
\label{prop-eqns-derived}
From the axioms of the calculus associated with meadow enriched \ACP\
process algebras, we can derive the equations from
Table~\ref{eqns-derived} for each binding term $p$ of sort $\Quant$,
binding term $P$ of sort $\Proc$, and $n,m \in \Natpos$.%
\begin{table}[!t]
\caption{Derived equations for comprehended terms}
\label{eqns-derived}
\begin{leqntbl}
\begin{ceqncol}
\Sum{1}{u} p  = p \subst{0}{u}                                        \\
\Sum{2}{u} p  = p \subst{0}{u} + p \subst{1}{u}                       \\
\Sum{2^{n+1}}{u} p =
\Sum{2}{u}
 \bigl(\Sum{2^{n}}{v} (p \subst{\ul{2} \mul v + u}{u})\bigr)          \\
\Sum{n+1}{u} p =
\Sum{2^m}{u} (\cond{0}{1 - \sign(u - \ul{n})}{p})
 & \mif n + 1 \leq 2^m
\eqnsep
\Prod{1}{u} p = p \subst{0}{u}                                        \\
\Prod{2}{u} p = p \subst{0}{u} \mul p \subst{1}{u}                    \\
\Prod{2^{n+1}}{u} p =
\Prod{2}{u}
 \bigl(\Prod{2^{n}}{v} (p \subst{\ul{2} \mul v + u}{u})\bigr)         \\
\Prod{n+1}{u} p =
\Prod{2^m}{u} (\cond{1}{1 - \sign(u - \ul{n})}{p})
 & \mif n + 1 \leq 2^m
\eqnsep
\Chc{1}{u} P  = P \subst{0}{u}                                        \\
\Chc{2}{u} P  = P \subst{0}{u} \altc P \subst{1}{u}                   \\
\Chc{2^{n+1}}{u} P =
\Chc{2}{u}
 \bigl(\Chc{2^{n}}{v} (P \subst{\ul{2} \mul v + u}{u})\bigr)          \\
\Chc{n+1}{u} P =
\Chc{2^m}{u} (\cond{\dead}{1 - \sign(u - \ul{n})}{P})
 & \mif n + 1 \leq 2^m
\eqnsep
\Seq{1}{u} P  = P \subst{0}{u}                                        \\
\Seq{2}{u} P  = P \subst{0}{u} \seqc P \subst{1}{u}                   \\
\Seq{2^{n+1}}{u} P =
\Seq{2}{u}
 \bigl(\Seq{2^{n}}{v} (P \subst{\ul{2} \mul v + u}{u})\bigr)          \\
\Seq{n+1}{u} P =
\Seq{2^m}{u} P \seqc \Seq{(n+1)-2^m}{u} (P \subst{\ul{2}^m + u}{u})
 & \mif 2^m < n + 1 < 2^{m+1}
\eqnsep
\Par{1}{u} P  = P \subst{0}{u}                                        \\
\Par{2}{u} P  = P \subst{0}{u} \parc P \subst{1}{u}                   \\
\Par{2^{n+1}}{u} P =
\Par{2}{u}
 \bigl(\Par{2^{n}}{v} (P \subst{\ul{2} \mul v + u}{u})\bigr)          \\
\Par{n+1}{u} P =
\Par{2^m}{u} P \parc \Par{(n+1)-2^m}{u} (P \subst{\ul{2}^m + u}{u})
 & \mif 2^m < n + 1 < 2^{m+1}
\end{ceqncol}
\end{leqntbl}
\end{table}
\end{proposition}
\begin{proof}
It follows immediately from the axioms for comprehended terms that the
first two equations for $\Sum{n}{}$ are derivable.
It is easy to prove by induction on $n$ that
\begin{ldispl}
\Sum{2 \mul n}{u} p =
\Sum{n}{u} (p \subst{2 \mul u}{u}) +
\Sum{n}{u} (p \subst{2 \mul u + 1}{u})
\end{ldispl}%
is derivable.
From this it follows easily that the third equation for $\Sum{n}{}$ is
derivable.
It is easy to prove by case distinction between $n = 1$ and $n > 1$ that
\begin{ldispl}
\Sum{n}{u} (\cond{0}{1 - \sign(u - \ul{0})}{p}) = p \subst{0}{u}
\end{ldispl}%
is derivable.
Using this fact, it is easy to prove by induction on $n$ that for all
$m \geq n + 1$:
\begin{ldispl}
\Sum{n+1}{u} p = \Sum{m}{u} (\cond{0}{1 - \sign(u - \ul{n})}{p})
\end{ldispl}%
is derivable.
From this it follows easily that the fourth equation for $\Sum{n}{}$ is
derivable.
The proofs for the equations for $\Prod{n}{}$, $\Chc{n}{}$,
$\Seq{n}{}$ and $\Par{n}{}$ go analogously, with the exception of the
fourth equation for $\Seq{n}{}$ and $\Par{n}{}$.
It is easy to prove by induction on $n$ that for all $m < n$:
\begin{ldispl}
\Seq{n}{u} P = \Seq{m}{u} P \seqc \Seq{n-m}{u} (P \subst{m + u}{u})
\end{ldispl}%
is derivable.
From this it follows easily that the fourth equation for $\Seq{n}{}$ is
derivable.
The proof for the fourth equation for $\Par{n}{}$ goes analogously.
\qed
\end{proof}

The axioms for comprehended terms give rise to a corollary about full
elimination of all variable-binding operators.
\begin{corollary}
\label{corollary-elim}
Let $t$ be a comprehended term without comprehended terms as proper
subterms, and let $k = \tsize(t)$.
Then there exists a term $t'$ over the signature of meadow enriched
\ACP\ process algebras such that $t = t'$ is derivable from the axioms
of the calculus associated with meadow enriched \ACP\ process algebras
and
\begin{itemize}
\item
$\tsize(t') = O(k^2 \mul 2^k)$;
\item
$\tsize(t') = \Omega(k \mul 2^{k-2})$ if $t$ is a term of the form
$\Sum{n}{u} t''$ or $\Prod{n}{u} t''$ and the number of times that $u$
occurs free in $t''$ is greater than zero;
\item
$\tsize(t') = \Omega(k \mul 2^{k-3})$ if $t$ is a term of the form
$\Chc{n}{u} t''$, $\Seq{n}{u} t''$ or $\Par{n}{u} t''$ and the number of
times that $u$ occurs free in $t''$ is greater than zero.
\end{itemize}
\end{corollary}
\begin{proof}
Term $t$ is a binding term of the form $\Bnd{n}{u} t''$, where
$\Bnd{n}{}$ is a variable-binding operator.
Let $k' = \tsize(t'')$,
let $k''$ be the number of times that $u$ occurs free in $t''$, and
let $l_i$ ($0 \leq i < n$) be the size of the smallest term $p$ over
the signature of meadow enriched \ACP\ process algebras such that
$p = \ul{i}$.
Then
$\tsize(t') = n \mul k' + \sum_{i = 0}^{n - 1} (k'' \mul l_i) + n - 1$.
Because $k = k' + \logii(n) + 1$, we know that $k' < k$, $\logii(n) < k$
and $n < 2^k$.
Moreover, we know that $k'' < k'$ and $l_i = \Theta(\logii(i + 1))$.
Hence $\tsize(t') = O(k^2 \mul 2^k)$.
We also know that
$k' \geq 1$ and, because $k = k' + \logii(n) + 1$,
$\logii(n) \geq k - 2$ and $n \geq 2^{k-2}$ if $t$ is of the form
$\Sum{n}{u} t''$ or $\Prod{n}{u} t''$; and that
$k' \geq 2$ and, because $k = k' + \logii(n) + 1$,
$\logii(n) \geq k - 3$ and $n \geq 2^{k-3}$ if $t$ is of the form
$\Chc{n}{u} t''$, $\Seq{n}{u} t''$ or $\Par{n}{u} t''$.
Hence, in the case where $k'' \geq 1$,
$\tsize(t') = \Omega(k \mul 2^{k-2})$ if $t$ is of the form
$\Sum{n}{u} t''$ or $\Prod{n}{u} t''$ and
$\tsize(t') = \Omega(k \mul 2^{k-3})$ if $t$ is of the form
$\Chc{n}{u} t''$, $\Seq{n}{u} t''$ or $\Par{n}{u} t''$.
\qed
\end{proof}

Proposition~\ref{prop-eqns-derived} gives rise to a corollary about
full elimination of all non-binary variable-binding operators.
\begin{corollary}
\label{corollary-elim-nonbin}
Let $t$ be a comprehended term without comprehended terms as proper
subterms, and let $k = \tsize(t)$.
Then there exists a binding term $t'$ without non-binary
variable-binding operators such that $t = t'$ is derivable from the
axioms of the calculus associated with meadow enriched \ACP\ process
algebras and
\begin{itemize}
\setlength{\itemsep}{.25ex}
\item
$\tsize(t') = O(k^3)$ if $t$ is a term of the form $\Sum{n}{u} t''$,
$\Prod{n}{u} t''$ or $\Chc{n}{u} t''$;
\item
$\tsize(t') = \Omega(k^2)$ if $t$ is a term of the form $\Sum{n}{u}
t''$, $\Prod{n}{u} t''$ or $\Chc{n}{u} t''$;
\item
$\tsize(t') = O(k^4)$ if $t$ is a term of the form $\Seq{n}{u} t''$
or~$\Par{n}{u} t''$;
\item
$\tsize(t') = \Omega(k^3)$ if $t$ is a term of the form $\Seq{n}{u} t''$
or~$\Par{n}{u} t''$ and the number of times that $u$ occurs free in 
$t''$ is greater than zero.
\end{itemize}
\end{corollary}
\begin{proof}
Firstly, we consider the case where $t$ is a term of the form
$\Sum{n}{u} t''$, $\Prod{n}{u} t''$ or $\Chc{n}{u} t''$.
Let $k' = \tsize(t'')$,
let $k''$ be the number of times that $u$ occurs free in $t''$, and
let $l'_n$ be the size of the smallest term $p$ over the signature of
meadow enriched \ACP\ process algebras such that
$p = 1 - \sign(u - \ul{n})$.
Then
$\tsize(t') = 
 k' +
 \smash{\sum_{i = 0}^{\logii(n)}} (k'' \mul (6 \mul i)) + \linebreak[2]
 \logii(n) \mul (\logii(n) + 1) + 4 \mul l'_n + 6$.
Because $k = k' + \logii(n) + 1$, we know that $k' < k$ and 
$\logii(n) < k$.
Moreover, we know that $k'' < k'$ and $l'_n = \Theta(\logii(n + 1))$.
Hence $\tsize(t') = O(k^3)$.
We also know that
$k' \geq 1$ and, because $k = k' + \logii(n) + 1$,
$\logii(n) \geq k - 2$ if $t$ is of the form $\Sum{n}{u} t''$ or
$\Prod{n}{u} t''$; and that
$k' \geq 2$ and, because $k = k' + \logii(n) + 1$,
$\logii(n) \geq k - 3$ if $t$ is of the form $\Chc{n}{u} t''$.
Hence, $\tsize(t') = \Omega(k^2)$.

Secondly, we consider the case where $t$ is a term of the form
$\Seq{n}{u} t''$ or $\Par{n}{u} t''$. 
Let $k' = \tsize(t'')$, and
let $k''$ be the number of times that $u$ occurs free in $t''$.
Then
$\tsize(t') \leq
 \smash{\sum_{i = 0}^{\logii(n)}}
  (k' +
   \smash{\sum_{j = 0}^{\logii(i)}} (k'' \mul (6 \mul j)) + 
   \logii(i) \mul (\logii(i) + 1))$.
Because $k = k' + \logii(n) + 1$, we know that $k' < k$ and
$\logii(n) < k$.
Moreover, we know that $k'' < k'$.
Hence $\tsize(t') = O(k^4)$.
We also have that
$\tsize(t') \geq
 k' +
 \smash{\sum_{i = 0}^{\logii(n)}} (k'' \mul (6 \mul i)) +
 \logii(n) \mul (\logii(n) + 1)$.
Because $k = k' + \logii(n) + 1$ and $k' \geq 2$, we also know that
$\logii(n) \geq k - 3$.
Hence, in the case where $k'' \geq 1$, $\tsize(t') = \Omega(k^3)$.
\qed
\end{proof}

Corollaries~\ref{corollary-elim} and~\ref{corollary-elim-nonbin} show
that much of the compactness that can be achieved with the
variable-binding operators of the calculus associated with meadow
enriched \ACP\ process algebras can already be achieved with the
binary variable-binding operators.

In Corollary~\ref{corollary-elim-nonbin}, $\tsize(t')$ is $O(k^4)$
instead of $O(k^3)$ if $t$ is of the form $\Seq{n}{u} t''$ or
$\Par{n}{u} t''$.
The origin of this is that \ACP\ process algebras do not have identity
elements for sequential and parallel composition.
In the setting of \ACP, the identity element for sequential composition,
as well as parallel composition, is known as the empty process.

\section{Adding an Identity Element for Sequential Composition}
\label{sect-ACP-process-alg-id-seqc}

In this section, we investigate the effect of adding an identity element
for sequential composition to \ACP\ process algebras on the result
concerning elimination of non-binary variable-binding operators
presented above.

The signature of these algebras is the signature of \ACP\ process
algebras extended with the following:
\begin{itemize}
\item
the \emph{empty process} constant $\const{\ep}{\Proc}$;
\item
the unary \emph{termination} operator $\funct{\termi}{\Proc}{\Proc}$.
\end{itemize}

Let $P$ be a closed term of sort $\Proc$.
Intuitively, the additional constant and operator can be explained as
follows:
\begin{itemize}
\item
$\ep$ is only capable of terminating successfully;
\item
$\termi(P)$ is only capable of terminating successfully if $P$ is
capable of terminating successfully and is not capable of doing anything
otherwise.
\end{itemize}

In the setting of \ACP, the addition of the empty process constant has
been treated in several ways.
The treatment in~\cite{KV85a} yields a non-associative parallel
composition operator.
The first treatment that yields an associative parallel composition
operator~\cite{Vra97a} is from 1986, but was not published until 1997.
The treatment in this paper is based on~\cite{BG87c}.

An \emph{ACP process algebra with an identity element for sequential
composition} is an algebra with the signature of \ACP\ process algebras
with an identity element for sequential composition that satisfies the
formulas given in Table~\ref{eqns-ACP} with the exception of
$x \parc y = (x \leftm y \altc y \leftm x) \altc x \commm y$
and the formulas given in Table~\ref{eqns-ep}.%
\begin{table}[!t]
\caption{Replacing and additional axioms for empty process constant}
\label{eqns-ep}
\begin{eqntbl}
\begin{eqncol}
x \seqc \ep = x                                                 \\
\ep \seqc x = x                                                 \\
x \parc y =
((x \leftm y \altc y \leftm x) \altc x \commm y) \altc
\termi(x) \seqc \termi(y)                                       \\
x \leftm \ep = x                                                \\
\ep \leftm x = \dead                                            \\
\ep \commm x = \dead                                            \\
\encap{H}(\ep) = \ep
\end{eqncol}
\qquad\quad
\begin{eqncol}
\termi(\ep) = \ep                                               \\
\termi(a) = \dead                                               \\
\termi(x \altc y) = \termi(x) \altc \termi(y)                   \\
\termi(x \seqc y) = \termi(x) \seqc \termi(y)                   \\
\termi(x) \seqc \termi(y) = \termi(y) \seqc \termi(x)           \\
x \altc \termi(x) = x
\end{eqncol}
\end{eqntbl}
\end{table}

We could dispense with the equations $a \leftm x = a \seqc x$ and
$a \commm b \seqc x = (a \commm b) \seqc x$ from Table~\ref{eqns-ACP}
because they have become derivable from the other equations.
In spite of the replacement of the equation
$x \parc y = (x \leftm y \altc y \leftm x) \altc x \commm y$
by the equation
$x \parc y =
 ((x \leftm y \altc\nolinebreak y \leftm x) \altc x \commm y) \altc
 \termi(x) \seqc \termi(y)$,
the equations characterizing \ACP\ process algebras with an identity
element for sequential composition constitute a conservative extension
of the equations characterizing \ACP\ process algebras.
The equation $\termi(x) \seqc \termi(y) = \termi(y) \seqc \termi(x)$ is
of importance because it makes the equation
$(x \parc y) \parc z = x \parc (y \parc z)$ derivable.
The equation $x \altc \termi(x) = x$ is of importance because it makes
the equation $x \parc \ep = x$ derivable.

Meadow enriched \ACP\ process algebras with an identity element for
sequential composition are defined like meadow enriched \ACP\ process
algebras.
We can associate a calculus with meadow enriched \ACP\ process algebras
with an identity element for sequential composition like we did before
for meadow enriched \ACP\ process algebras.

By the addition of an identity element for sequential composition, the
properties of $\Seq{n}{}$ and $\Par{n}{}$ with respect to elimination of
non-binary variable-binding operators become comparable to the
properties of $\Sum{n}{}$, $\Prod{n}{}$ and $\Chc{n}{}$ with respect to
elimination of non-binary variable-binding operators.
\begin{proposition}
\label{prop-eqns-derived-ep}
From the axioms of the above-mentioned calculus, we can derive the
following equations  for each binding term $P$ of sort $\Proc$ and
$n,m \in \Natpos$:
\begin{ldispl}
\Seq{n+1}{u} P =
\Seq{2^m}{u} (\cond{\ep}{1 - \sign(u - \ul{n})}{P})
\hfill \mif n + 1 \leq 2^m\;,
\\[.75ex]
\Par{n+1}{u} P =
\Par{2^m}{u} (\cond{\ep}{1 - \sign(u - \ul{n})}{P})
\quad \mif n + 1 \leq 2^m\;.
\end{ldispl}%
\end{proposition}
\begin{proof}
The proofs for these equations go analogously to the proofs for the last
equations for $\Sum{n}{}$, $\Prod{n}{}$ and $\Chc{n}{}$ in the proof of
Proposition~\ref{prop-eqns-derived}.
\qed
\end{proof}

Proposition~\ref{prop-eqns-derived-ep} gives rise to a corollary about
full elimination of the non-binary variable-binding operators for
sequential and parallel composition in the presence of an identity
element for sequential composition.
\begin{corollary}
\label{corollary-elim-nonbin-ep}
Let $t$ be a comprehended term of the form $\Seq{n}{u} t''$ or
$\Par{n}{u} t''$ without comprehended terms as proper subterms, and let
$k = \tsize(t)$.
Then there exists a binding term $t'$ without non-binary
variable-binding operators such that $t = t'$ is derivable from the
axioms of the above-mentioned calculus and $\tsize(t') = O(k^3)$ and
$\tsize(t') = \Omega(k^2)$.
\end{corollary}
\begin{proof}
The proof goes analogously to the case where $t$ is of the form
$\Sum{n}{u} t''$, $\Prod{n}{u} t''$ or $\Chc{n}{u} t''$ in the proof of
Corollary~\ref{corollary-elim-nonbin}.
\qed
\end{proof}

Corollaries~\ref{corollary-elim-nonbin} 
and~\ref{corollary-elim-nonbin-ep} imply that, on full elimination of 
the non-binary variable-binding operators for sequential and parallel 
composition, the addition of an identity element for sequential 
composition to \ACP\ process algebras gives rise to polynomially 
smaller terms.

\section{Adding Process Sequences}
\label{sect-ACP-process-alg-proc-vec}

In this section, we introduce process sequences to demonstrate that
there is an alternative to introducing variable-binding operators for
several associative operators on processes.

The signature of \ACP\ process algebras with an identity element for
sequential composition and process sequences is the signature of \ACP\
process algebras with an identity element for sequential composition
extended with the sort $\PS$ of \emph{process sequences} and the
following constants and operators:
\begin{itemize}
\item
the \emph{empty process sequence} constant
$\const{\emptyseq}{\PS}$;
\item
the unary \emph{singleton process sequence} operator
$\funct{\seq{\ph}}{\Proc}{\PS}$;
\item
the binary \emph{process sequence concatenation} operator
$\funct{\conc}{\PS \x \PS}{\PS}$;
\item
the unary \emph{generalized alternative composition} operator
$\funct{\sChc}{\PS}{\Proc}$;
\item
the unary \emph{generalized sequential composition} operator
$\funct{\sSeq}{\PS}{\Proc}$;
\item
the unary \emph{generalized parallel composition} operator
$\funct{\sPar}{\PS}{\Proc}$.
\end{itemize}
We assume that there is a countably infinite set $\cV$ of variables
of sort $\PS$, which contains $\alpha$, $\beta$ and $\gamma$, with and
without subscripts.
We use the same notational conventions as before.
In addition, we use infix notation for the binary operator $\conc$ and
mixfix notation for the unary operator $\seq{\ph}$.

The constant and the first two operators introduced above are the usual
ones for sequences, which gives an appropriate intuition about them.
The remaining three operators introduced above generalize alternative,
sequential and parallel composition to an arbitrary finite number of
processes.

An \emph{ACP process algebra with an identity element for sequential
composition and process sequences} is an algebra with the signature of
\ACP\ process algebras with an identity element for sequential
composition and process sequences that satisfies the formulas given in
Table~\ref{eqns-ACP} with the exception of
$x \parc y = (x \leftm y \altc y \leftm x) \altc x \commm y$
and the formulas given in Tables~\ref{eqns-ep} and~\ref{eqns-ps}.%
\begin{table}[!t]
\caption{Additional axioms for process sequences}
\label{eqns-ps}
\begin{eqntbl}
\begin{eqncol}
\alpha \conc \emptyseq = \alpha                                       \\
\emptyseq \conc \alpha = \alpha                                       \\
(\alpha \conc \beta) \conc \gamma = \alpha \conc (\beta \conc \gamma)
\eqnsep
\sChc(\emptyseq) = \dead                                              \\
\sChc(\seq{x}) = x                                                    \\
\sChc(\seq{x} \conc \alpha) = x \altc \sChc(\alpha)
\end{eqncol}
\qquad\quad
\begin{eqncol}
\sSeq(\emptyseq) = \epsilon                                           \\
\sSeq(\seq{x}) = x                                                    \\
\sSeq(\seq{x} \conc \alpha) = x \seqc \sSeq(\alpha)
\eqnsep
\sPar(\emptyseq) = \epsilon                                           \\
\sPar(\seq{x}) = x                                                    \\
\sPar(\seq{x} \conc \alpha) = x \parc \sPar(\alpha)
\end{eqncol}
\end{eqntbl}
\end{table}

If we would introduce process sequences in the absence of an identity
element for sequential composition, we should consider non-empty process
sequences only.

Meadow enriched \ACP\ process algebras with an identity element for
sequential composition and process sequences are defined like meadow
enriched \ACP\ process algebras.
We can associate a calculus with meadow enriched \ACP\ process algebras
with an identity element for sequential composition and process
sequences like we did before for meadow enriched \ACP\ process algebras.
Moreover, we can extend the resulting calculus with variable-binding
operators that generalize the process sequence concatenation operator.
For the terms of the extended calculus, we need the following additional
formation rule:
\begin{itemize}
\item
if $u \in \cU$ and $t \in \BT{\PS}$, then, for each $n \in \Natpos$,
$\Conc{n}{u} t \in \BT{\PS}$.
\end{itemize}

The axioms of the extended calculus are the formulas given in
Tables~\ref{eqns-meadow}--\ref{eqns-compr-terms}
and~\ref{eqns-ep}--\ref{eqns-compr-terms-ps}.%
\begin{table}[!t]
\caption{Additional axioms for comprehended terms of sort $\PS$}
\label{eqns-compr-terms-ps}
\begin{eqntbl}
\begin{eqncol}
\Conc{n}{u} S = \Conc{n}{v} (S \subst{v}{u})                          \\
\Conc{1}{u} S = S \subst{0}{u}                                        \\
\Conc{n+1}{u} S =
S \subst{0}{u} \conc \Conc{n}{u} (S \subst{u + 1}{u})
\end{eqncol}
\end{eqntbl}
\end{table}
Like some equations in Tables~\ref{eqns-ACP}--\ref{eqns-compr-terms},
the equations in Table~\ref{eqns-compr-terms-ps} are actually schemas of
equations: $S$ is a syntactic variable which stands for an arbitrary
binding term of sort $\PS$, and $n$ stands for an arbitrary positive
natural number.

The properties of $\Conc{n}{}$ with respect to elimination of non-binary
variable-binding operators are comparable to the properties of
$\Chc{n}{}$, $\Seq{n}{}$ and $\Par{n}{}$ with respect to elimination of
non-binary variable-binding operators.
\begin{proposition}
\label{prop-eqns-derived-ps}
From the axioms of the extended calculus, we can derive the following
equations for each binding term $S$ of sort $\PS$ and $n,m \in \Natpos$:
\begin{ldispl}
\Conc{1}{u} S = S \subst{0}{u}\;,
\\[.35ex]
\Conc{2}{u} S = S \subst{0}{u} \conc S \subst{1}{u}\;,
\\
\Conc{2^{n+1}}{u} S =
\Conc{2}{u}
 {\bigl(\Conc{2^{n}}{v} (S \subst{\ul{2} \mul v + u}{u})}\bigr)\;,
\\
\Conc{n+1}{u} S =
\Conc{2^m}{u} (\cond{\emptyseq}{1 - \sign(u - \ul{n})}{S})
\quad \mif n + 1 \leq 2^m\;.
\end{ldispl}%
\end{proposition}
\begin{proof}
The proof goes analogously to the case of the equations for $\Sum{n}{}$
in the proof of Proposition~\ref{prop-eqns-derived}.
\qed
\end{proof}

Proposition~\ref{prop-eqns-derived-ps} gives rise to a corollary about
full elimination of the non-binary variable-binding operators for
process sequence concatenation.
\begin{corollary}
\label{corollary-elim-nonbin-ps}
Let $t$ be a comprehended term of the form $\Conc{n}{u} t''$ without
comprehended terms as proper subterms, and let $k = \tsize(t)$.
Then there exists a binding term $t'$ without non-binary
variable-binding operators such that $t = t'$ is derivable from the
axioms of the extended calculus and $\tsize(t') = O(k^3)$ and
$\tsize(t') = \Omega(k^2)$.
\end{corollary}
\begin{proof}
The proof goes analogously to the case where $t$ is of the form
$\Sum{n}{u} t''$, $\Prod{n}{u} t''$ or $\Chc{n}{u} t''$ in the proof of
Corollary~\ref{corollary-elim-nonbin}.
\qed
\end{proof}

In the presence of the operators $\sChc$, $\sSeq$ and $\sPar$ and the
variable-binding operator $\Conc{n}{}$, the variable-binding operators
$\Chc{n}{}$, $\Seq{n}{}$, and $\Par{n}{}$ are superfluous.
\begin{proposition}
\label{prop-alt-binders}
From the axioms of the extended calculus, we can derive the following
equations for each binding term $P$ of sort $\Proc$ and $n \in \Natpos$:
\begin{ldispl}
\Chc{n}{u} P = \sChc \bigl(\Conc{n}{u} \seq{P}\bigr)\;, \qquad
\Seq{n}{u} P = \sSeq \bigl(\Conc{n}{u} \seq{P}\bigr)\;, \qquad
\Par{n}{u} P = \sPar \bigl(\Conc{n}{u} \seq{P}\bigr)\;.
\end{ldispl}%
\end{proposition}
\begin{proof}
This is easy to prove by induction on $n$.
\qed
\end{proof}

If we would introduce quantity sequences as well, we could get a similar
result for the variable-binding operators $\Sum{n}{}$ and $\Prod{n}{}$.

Proposition~\ref{prop-alt-binders} shows that there is an alternative to 
introducing variable-binding operators for alternative, sequential and 
parallel composition.
However, this proposition also gives rise to a corollary about full 
elimination of the non-binary variable-binding operators for 
alternative, sequential and parallel composition.
\begin{corollary}
\label{corollary-elim-nonbin-alt}
Let $t$ be a comprehended term of the form $\Chc{n}{u} t''$, 
$\Seq{n}{u} t''$ or $\Par{n}{u} t''$ without comprehended terms as proper 
subterms, and let $k = \tsize(t)$.
Then there exists a binding term $t'$ without non-binary variable-binding 
operators such that $t = t'$ is derivable from the axioms of the extended 
calculus and $\tsize(t') = O(k^3)$ and $\tsize(t') = \Omega(k^2)$.
\end{corollary}
\begin{proof}
This is a direct consequence of Corollary~\ref{corollary-elim-nonbin-ps}
and Proposition~\ref{prop-alt-binders}.
\qed
\end{proof}

Corollary~\ref{corollary-elim-nonbin-alt} implies that in the presence 
of an identity element for sequential composition, on full elimination 
of the non-binary variable-binding operators for alternative, sequential 
and parallel composition, the addition of process sequences to \ACP\ 
process algebras does not give rise to significantly smaller or larger 
terms.

\section{Concluding Remarks}
\label{sect-conclusions}

We have introduced the notion of an \ACP\ process algebra.
The set of equations that have been taken to characterize \ACP\ process
algebras is a revision of the axiom system \ACP.
We consider this revision worth mentioning of itself, if only because
it removes the need to have a constant for each atomic action.
We have also introduced the notion of a meadow enriched \ACP\ process
algebra.
This notion is a simple generalization of the notion of an \ACP\ process
algebra to processes in which data are involved, the mathematical
structure of data being a meadow.
The primary mathematical structure for calculations is unquestionably a
field, and a meadow differs from a field only in that the multiplicative
inverse operation is made total by imposing that the multiplicative
inverse of zero is zero.
Therefore, we consider the combination of \ACP\ process algebras and
meadows made in this paper, a combination with potentially many 
applications.

For all associative operators from the signature of meadow enriched
\ACP\ process algebras that are not of an auxiliary nature, we have 
introduced variable-binding operators as generalizations.
Thus, we have obtained a process calculus whose terms can be interpreted
in all meadow enriched \ACP\ process algebras.
We have shown that the use of variable-binding operators that bind
variables with a two-valued range can already have a major impact on the
size of terms, and that the impact can be further increased if we add an
identity element for sequential composition to meadow enriched \ACP\
process algebras.
In addition, we have demonstrated that there is an alternative to 
introducing variable-binding operators for several associative operators
on processes if we add a sort of process sequences and suitable 
operators on process sequences to meadow enriched \ACP\ process 
algebras.

All variable-binding operators of the calculus associated with meadow
enriched \ACP\ process algebras can be eliminated from all terms of the
calculus by means of its axioms, and all terms of the calculus can be
directly interpreted in meadow enriched \ACP\ process algebras.
Therefore, although they yield a calculus, we consider these
variable-binding operators to constitute a process algebraic feature.
Fitting them in an algebraic framework does not involve any serious
theoretical complication.

Different from the variable-binding operators introduced in this paper,
the variable-binding operators from $\mu$CRL and PSF that generalize
associative operators of \ACP\ do not give rise to finitary comprehended
terms.
It is much more difficult to fit the variable-binding operators from
those formalisms in an algebraic framework, see e.g.~\cite{Lut02a}.
This also holds for the integration operator, which is found in
extensions of the axiom system \ACP\ concerning timed processes to allow
for the alternative composition of a continuum of differently timed
processes to be expressed (see e.g.~\cite{BM02a}).
It is worth mentioning that in effective $\mu$CRL, a restriction of 
$\mu$CRL for which a simulator is feasible (see e.g.~\cite{GP95a}), the 
variable bound by the variable binding operator that generalizes 
alternative composition must have a finite range.

We have also attempted to fit variable-binding operators that bind
variables with an infinite range in an algebraic framework.
We have looked at binding algebras~\cite{Sun99a}, which are second-order
algebras of a specific kind that covers variable-binding operators.
The problem is that the theory of binding algebras is insufficiently
elaborate for our purpose.
For example, it is not known whether the important characterization
results from the theory of first-order algebras, i.e.\ Birkhoff's
variety result and Malcev's quasi-variety result (see
e.g.~\cite{BS81a,MT92a}), have generalizations for binding algebras.

It is known that many important results from the theory of first-order
algebras, including the above-mentioned ones, have generalizations for
higher-order algebras as considered in the theory of general
higher-order algebras developed in~\cite{Mei92a,KM96a,Mei03a}.
Therefore, we have also considered the replacement of variable-binding
operators by higher-order operators that give rise to such higher-order
algebras.
However, owing to the absence of bound variables, additional
higher-order operators are needed which serve the same purpose as the
combinators of combinatory logic~\cite{HS86a}.
Thus, this leads to the line taken earlier with combinatory process
algebra~\cite{BBP94b}.

\bibliographystyle{splncs03}
\bibliography{PA}

\end{document}